\documentclass[11pt,a4paper]{article}

\usepackage{amsmath,amsfonts,amssymb,amsthm}
\usepackage{amsthm}
\usepackage{graphicx,color}
\usepackage{boxedminipage}
\usepackage[ruled,boxed,linesnumbered]{algorithm2e}
\usepackage{framed}
\usepackage{thmtools}
\usepackage{thm-restate}
\usepackage{xspace}
\usepackage{todonotes}
\usepackage{footnote}
\usepackage{mathtools}
\usepackage{mathrsfs}
\usepackage{vmargin}
\setmarginsrb{1in}{1in}{1in}{1in}{0pt}{0pt}{0pt}{6mm}
\usepackage{todonotes}
\usepackage{footnote}
\usepackage{rotating}
 \usepackage[pdftex, plainpages = false, pdfpagelabels, 
                 bookmarks=false,
                 bookmarksopen = true,
                 bookmarksnumbered = true,
                 breaklinks = true,
                 linktocpage,
                 pagebackref,
                 colorlinks = true,  % was true
                 linkcolor = blue,
                 urlcolor  = blue,
                 citecolor = red,
                 anchorcolor = green,
                 hyperindex = true,
                 hyperfigures
                 ]{hyperref} 
 \usepackage{xifthen}
 \usepackage{tabularx}
\usepackage{tikz} 
 \usetikzlibrary{calc}
 \usepackage{thm-autoref}
\usepackage[nameinlink]{cleveref}

\newtheorem{theorem}{Theorem}
\newtheorem{lemma}{Lemma}
\newtheorem{claim}{Claim}[section]

\newtheorem{observation}{Observation}

\theoremstyle{definition}
\newtheorem{reduction}{Reduction Rule}[section]

\newcommand{\dist}{\textrm{\rm dist}}

\DeclareMathOperator{\operatorClassFPT}{FPT\xspace}
\newcommand{\classFPT}{\ensuremath{\operatorClassFPT}\xspace}

\newlength{\RoundedBoxWidth}
\newsavebox{\GrayRoundedBox}
\newenvironment{GrayBox}[1]%
   {\setlength{\RoundedBoxWidth}{.93\textwidth}
    \def\boxheading{#1}
    \begin{lrbox}{\GrayRoundedBox}
       \begin{minipage}{\RoundedBoxWidth}}%
   {   \end{minipage}
    \end{lrbox}
    \begin{center}
    \begin{tikzpicture}%
       \node(Text)[draw=black!20,fill=white,rounded corners,%
             inner sep=2ex,text width=\RoundedBoxWidth]%
             {\usebox{\GrayRoundedBox}};
        \coordinate(x) at (current bounding box.north west);
        \node [draw=white,rectangle,inner sep=3pt,anchor=north west,fill=white] 
        at ($(x)+(6pt,.75em)$) {\boxheading};
    \end{tikzpicture}
    \end{center}}

\newenvironment{defproblemx}[2][]{\noindent\ignorespaces%
                                \FrameSep=6pt%
                                \parindent=0pt%
                \vspace*{-1.5em}
                \ifthenelse{\isempty{#1}}{%
                  \begin{GrayBox}{\textsc{#2}}%                
                }{%
                  \begin{GrayBox}{\textsc{#2} parameterized by~{#1}}%  
                }
                \begin{tabular*}{\textwidth}{@{\hspace{.1em}} >{\itshape} p{1.8cm} p{0.8\textwidth} @{}}%        
            }{
                \end{tabular*}%
                \end{GrayBox}%
                \ignorespacesafterend
            }

\newcommand{\Oh}{\mathcal{O}}

\newcommand{\pname}{\textsc}
\newcommand{\ProblemFormat}[1]{\pname{#1}}
\newcommand{\ProblemIndex}[1]{\index{problem!\ProblemFormat{#1}}}
\newcommand{\ProblemName}[1]{\ProblemFormat{#1}\ProblemIndex{#1}{}\xspace}

\newcommand{\probTB}{\ProblemName{Telephone Broadcast}} %Problem

\begin{document}

\title{Parameterized Complexity of Broadcasting in Graphs\thanks{The research leading to these results has received funding from the Research Council of Norway via the project  BWCA (grant no. 314528).}
}

\author{
 Fedor V. Fomin\thanks{Department of Informatics, University of Bergen, Norway}
    \and
    Pierre Fraigniaud\thanks{Institut de Recherche en Informatique Fondamentale, Universit\'e Paris Cité and CNRS, France. Additional support from the ANR project DUCAT. }\addtocounter{footnote}{-2}   
     \and
    Petr A. Golovach\footnotemark{}
}

\date{}

\maketitle

\begin{abstract}
The task of the broadcast problem is, given a graph~$G$ and a source vertex~$s$, to compute the minimum number of rounds required to disseminate a piece of information from~$s$ to all vertices in the graph. It is assumed that, at each round, an informed vertex can transmit the information to at most one of its neighbors. The broadcast problem is known to NP-hard.  We show that the problem is FPT when parametrized by the size~$k$ of a feedback edge-set, or by the size~$k$ of a vertex-cover, or by $k=n-t$, where $t$ is the input deadline for the broadcast protocol to complete.  
\end{abstract}

\section{Introduction}\label{sec:intro}

The aim of \emph{broadcasting} in a network is to transmit a message from a given source node of the network to all the other nodes. Let $G=(V,E)$ be a connected simple graph modeling the network, and let $s\in V$ be the source of a message~$M$. The standard \emph{telephone model}~\cite{HedetniemiHL88} assumes that the communication proceeds in synchronous rounds. At any given round, any node $u\in V$ aware of~$M$ can forward~$M$ to at most one neighbor $v$ of~$u$.  The minimum number of rounds for broadcasting a message from $s$ in~$G$ to all other vertices is denoted by $b(G,s)$, and we let $b(G)=\max_{s\in V}b(G,s)$ be the broadcast time of graph~$G$. As the number of informed nodes (i.e., nodes aware of the message) can at most double at each round, $b(G,s)\geq \lceil \log_2n\rceil$ in $n$-node networks. On the other hand, since $G$ is connected, at least one uninformed node receives the message at any given round, and thus $b(G)\leq n-1$. Both bounds are tight, as witnessed by the complete graph~$K_n$ and the path~$P_n$, respectively. The problem of computing the broadcast time $b(G,s)$ for a given graph~$G$ and a given source~$s\in V$ is NP-hard~\cite{SlaterCH81}. Also, the results of~\cite{PapadimitriouY82} imply that  it is NP-complete to decide whether $b(G,s)\leq t$ for graphs with $n=2^t$ vertices.

Three lines of research have emerged since the early days of studying broadcasting in the telephone model. One line is devoted to determining the broadcast time of specific classes of graphs judged important for their desirable properties as interconnection networks (e.g., hypercubes, de Bruijn graphs, Cube Connected Cycles, etc.). We refer to the surveys~\cite{FraigniaudL94,Hromkovic05} for this matter. Another line of research takes inspiration from extremal graph theory. It aims at constructing $n$-node graphs~$G$ with optimal broadcast time~$b(G)=\lceil\log_2n\rceil$ and minimizing the number of edges sufficient to guarantee this property. Let $B(n)$ be the minimum number of edges of $n$-node graphs with broadcast time $\lceil\log_2n\rceil$. It is known~\cite{GrigniP91} that $B(n)=\Theta(n\,L(n))$ where $L(n)$ denotes the number of consecutive leading 1s in the binary representation of~$n-1$. On the other hand, it is still not known whether $B(\cdot)$ is non-decreasing for  $2^t\leq n < 2^{t+1}$, for every $t\in \mathbb{N}$. We are interested in a third, more recent line of research, namely the design of algorithms computing efficient broadcast protocols. Note that a protocol for broadcasting from a source~$s$ in a graph~$G$ can merely be represented as a spanning tree~$T$ rooted at~$s$, with an ordering of all the children of each node in the tree. 

Polynomial-time algorithms are known for trees~\cite{SlaterCH81} and some classes of tree-like graphs~\cite{CevnikZ17,GholamiHM23,HarutyunyanM08}.
 Several (polynomial-time) approximation algorithms have been designed for the broadcast problem. In particular, the algorithm in~\cite{KortsarzP95} computes, for every graph~$G$ and every source~$s$, a broadcast protocol from~$s$ performing in $2\,b(G,s)+O(\sqrt{n})$ rounds, hence this algorithm has approximation ratio $2+o(1)$ for graphs with broadcast time $\gg \sqrt{n}$, but $\tilde{\Theta}(\sqrt{n})$ in general. Later, a series of papers tighten the approximation ratio, from $O(\log^2n/\log\log n)$~\cite{Ravi94}, to $O(\log n)$~\cite{Bar-NoyGNS00}, and eventually $O(\log n/\log\log n)$~\cite{ElkinK06}, which is, up to our knowledge the current best approximation ratio for the broadcast problem. Better approximation ratios can be obtained for specific classes of graphs~\cite{BhabakH15,BhabakH19,HarutyunyanH23}.

Despite all the achievements obtained on the broadcast problem, it has not yet been approached from the parameterized complexity viewpoint~\cite{CyganFKLMPPS15}. There might be a good reason for that. Since at most $2^t$ vertices can have received the message after $t$ communication rounds, an instance of the broadcast problem with time-bound $t$ in an $n$-vertex graph is a no-instance whenever $n>2^t$. It follows that the broadcast problem has a trivial kernel when parameterized by the broadcast time. This makes the natural parameterization by the broadcast time not very significant. Nevertheless, as we shall show in this paper, there is a  parameterization below the natural upper bound for the number of rounds that leads to interesting conclusions.   

\subsection{Our Results}

Let \probTB be the following problem: given a connected graph $G=(V,E)$, a source vertex $s\in V$, and a nonnegative integer~$t$, decide whether there is a broadcast protocol from~$s$ in~$G$ that ensures that all the vertices of~$G$ get the message in at most~$t$ rounds.  We first show that \probTB can be solved in a single-exponential time by an exact algorithm. 

\begin{theorem}\label{thm:exact}
\probTB can be solved in $3^n\cdot n^{\Oh(1)}$ time for $n$-vertex graphs.
\end{theorem}

Motivated by the fact that the complexity of \probTB  remains open in pretty simple tree-like graphs  (e.g., cactus graphs, and therefore outerplanar graphs), we first consider the \emph{cyclomatic}  number as a parameter, i.e., the minimum size of a feedback edge set, that is, the size of the smallest set of edges whose deletion results to an acyclic graph. We show that \probTB is \classFPT when parameterized by this parameter. 

\begin{theorem}\label{thm:cyc}
\probTB can be solved in $2^{\Oh(k\log k)}\cdot n^{\Oh(1)}$ time for $n$-vertex graphs with cyclomatic number at most~$k$. 
\end{theorem}

As far as we know, no NP-hardness result is known on graphs of treewidth at most~$k\geq 2$. While we did not progress in that direction, we  provide an interesting result for a stronger parameterization, namely the \emph{vertex cover} number of a graph. (Note that, for all graphs, the treewidth never exceeds the vertex cover number.)  As for the cyclomatic number, we do not only show that, for every fixed~$k$, the broadcast time can be found in polynomial time on graphs with vertex cover at most~$k$, but we prove a stronger result: the problem is FPT. 

\begin{theorem}\label{thm:vc}
\probTB can be solved in $2^{\Oh(k2^k)}\cdot n^{\Oh(1)}$ time for $n$-vertex graphs with a vertex cover of size at most $k$. 
\end{theorem}

Finally, we focus on graphs with very large broadcast time, for which the algorithm in~\cite{KortsarzP95} provides hope to derive very efficient broadcast protocols as this algorithm constructs a broadcast protocol performing in $2\,b(G)+O(\sqrt{n})$ rounds. While we were not able to address the problem over the whole range $\sqrt{n}\ll t \leq n-1$, we were able to provide answers for the range $n-O(1)\leq t \leq n-1$. More specifically, we consider the parameter $k=n-t$, and study the kernelization for the problem under such parametrization. 

\begin{theorem}\label{thm:kernel}
\probTB  admits a kernel with $\Oh(k)$ vertices in $n$-vertex graphs when parameterized by $k=n-t$. 
\end{theorem}

As a direct consequence of Theorem~\ref{thm:kernel}, \probTB  is FPT for the parameterization by $k=n-t$. Specifically the problem can be solved in $2^{O(k)}\cdot n^{O(1)}$ time.

\subsection{Related Work}

A classical generalization of the broadcast problem is the \emph{multicast} problem, in which the message should only reach a given subset of target vertices in the input graph. Many of the previously mentioned approximation algorithms for the broadcast problem extend to the multicast problem, and, in particular, the algorithm in~\cite{ElkinK06} is an $O(\log k/\log\log k)$-approximation algorithm for the multicast problem with $k$ target nodes. 

Many variants of the telephone model have been considered in the literature, motivated by different network technologies. One typical example is the \emph{line} model~\cite{Farley80}, in which a call between a vertex $u$ and a vertex $v$ is implemented by a path between $u$ and $v$ in the graph, with the constraint that all calls performed at the same round must be performed along edge-disjoint paths. (The intermediate nodes along the path do \emph{not} receive the message, which ``cut through'' them.)  Interestingly, the broadcast time of \emph{every} $n$-node graph is exactly $\lceil\log_2n\rceil$. The result extends to networks in which the paths are constructed by an underlying routing function~\cite{CohenFKR98}. The vertex-disjoint variant of the line model, i.e., the line model in which the calls performed at the same round must take place along vertex-disjoint paths, is significantly more complex. There is an $O(\log n/\log\log n)$-approximation algorithm for the vertex-disjoint line model~\cite{KortsarzP95}, which naturally extend to an $O(\log n/\log \mbox{\footnotesize OPT})$-approximation algorithm --- see also~\cite{Fraigniaud01} where an explicit $O(\log n/\log\mbox{\footnotesize OPT})$-approximation algorithm is provided. It is also worth mentioning that the broadcast model has been also extensively studied in models aiming at capturing any type of node- or link-latencies, e.g., the message takes $\lambda_e$ units of time to traverse edge~$e$, and the algorithm in~\cite{Bar-NoyGNS00} also handles  such constraints. Other variants take into account the size of the message, e.g., a message of $L$~bits takes time $\alpha+\beta\cdot L$ to traverse an edge (see~\cite{JohnssonH89}). Under such a model, it might be efficient to split the original message into smaller packets and pipeline the broadcast of these packets through disjoint spanning trees~\cite{JohnssonH89,StoutW90}.

\section{Preliminaries}\label{sec:prelim} 

\paragraph{Parameterized Complexity.}
We refer to the book of Cygan et  al.~\cite{CyganFKLMPPS15} for a detailed introduction to Parameterized Complexity and give here only crucial definitions. 

Formally, a \emph{parameterized problem} is language $L\subseteq \Sigma^*\times\mathbb{N}$, where $\Sigma^*$ is the set of strings over a finite alphabet $\Sigma$. Hence, an instance of a parameterized problem is a pair $(x,k)$, where  $x\in\Sigma^*$ is a string encoding the input and $k\in\mathbb{N}$ is a \emph{parameter}.  
A parameterized problem is said to be \emph{fixed-parameter tractable} (or \classFPT) if it can be solved in $f(k)\cdot |x|^{\Oh(1)}$ time for some computable function~$f(\cdot)$. 

A \emph{kernelization } algorithm (or simply a \emph{kernel}) for a parameterized problem $L$ is a polynomial algorithm that maps each instance $(x,k)$ of $L$ into an instance $(x',k')$ of the same problem such that
\begin{itemize}
\item[(i)] $(x,k)\in L$ if and only if $(x',k')\in L$, that is, the instances $(x,k)$ and $(x',k')$ are equivalent, and
\item[(ii)] $|x'|+k'\leq g(k)$ for a computable function~$g(\cdot)$.
\end{itemize}
It is said that $g(\cdot)$ is the \emph{kernel size}, and a kernel is \emph{polynomial} if $g(\cdot)$ is polynomial. 
It is well-known that every decidable parameterized problem is \classFPT if and only if it has a kernel. However, there are \classFPT parameterized problems that do not admit polynomial kernels up to some reasonable complexity assumptions.  
It is common to present a kernelization algorithm as a series of \emph{reduction rules}. A reduction rule for a parameterized problem is an algorithm that takes an instance of the problem and computes in polynomial time another instance that is more ``simple'' in a certain way.
A reduction rule is \emph{safe} (or \emph{sound}) if the computed instance is equivalent to the input instance.

\paragraph{Integer Programming.}

We will use integer linear programming as a subroutine. The task of \textsc{$p$-Variable Integer Linear Programming Feasibility} problem is to decide, given an $m\times p$ matrix $A$ over $\mathbb{Z}$ and a vector $b\in \mathbb{Z}^{m}$, whether there is a vector $x\in \mathbb{Z}^{p}$ such that $Ax\leq b$.  It was proved by Lenstra~\cite{Lenstra83} and Kannan~\cite{Kannan87}
that this problem is \classFPT when parameterized by $p$ and these results were improved by  Frank and Tardos~\cite{FrankT87}.	

\begin{lemma}[\cite{FrankT87,Kannan87,Lenstra83}]\label{lem:ILP}
\textsc{$p$-Variable Integer Linear Programming Feasibility} can be solved using $\Oh(p^{2.5p+o(p)}\cdot L)$ arithmetic operations and space polynomial in $L$, where $L$ is the number of bits in the input.
\end{lemma}

\paragraph{Graphs.}
We use standard graph-theoretic notation and refer to the textbook of Diestel~\cite{Diestel12} for non-defined notions. We consider only finite,  undirected graphs. We use $V(G)$ and $E(G)$ to denote the sets of vertices and edges of a graph $G$.
  We use $n$ and $m$ to denote the number of vertices and edges  if this does not create confusion. 
For a  graph $G$ and a subset $X\subseteq V(G)$ of vertices, we write $G[X]$ to denote the subgraph of $G$ induced by $X$.
We use $G-X$ to denote the  graph obtained by deleting the vertices of $X$, that is, $G-X=G[V(G)\setminus X]$; we write $G-v$ instead of $G-\{v\}$ for a single element set.
We use the same convention for edges and write $G-S$ and $G-e$ for the graph obtained from $G$ by the removal of a set of edges $S$ and a single edge $e$, respectively.
 For a vertex $v$, $N_G(v)=\{u\in V(G)\mid vu\in E(G)\}$ is the \emph{open neighborhood} of $v$, and $d_G(v)=|N_G(v)|$ is the \emph{degree} of $v$. For a set $S\subseteq V(G)$, $N_G(S)=\bigcup_{v\in S}N_G(v)\setminus S$. Two vertices $u$ and $v$ are \emph{(false) twins} if $N_G(u)=N_G(v)$.
It is said that $G'=G/e$ is obtained by the \emph{contraction} of an edge $e=uv\in E(G)$ if $G'$ is obtained by removing $u$ and $v$ and replacing them by a single vertex adjacent to all the vertices of $N_G(\{u,v\})$. Recall that in \probTB we have a specific source vertex $s$. To keep such a vertex under contraction, we assume that the new vertex obtained by the contraction of an edge incident to $s$ becomes the source and we use the same name $s$ for it.

We write $P=v_1\cdots v_k$ to denote a \emph{path} with the vertices $v_1,\ldots,v_k$ and edges 
$v_1v_2,\ldots,v_{k-1}v_k$;
$v_1$ and $v_k$ are the \emph{end-vertices} of $P$ and the vertices $v_2,\dots,v_{k-1}$ are \emph{internal}. We consider only simple paths, that is, the vertices $v_1,\ldots,v_k$ are distinct. We say that $P$ is an \emph{$(v_1,v_k)$-path}. The \emph{length} of $P$ is the number of edges in $P$, and the \emph{distance} $\dist_G(u,v)$ between two vertices $u$ and $v$ is the minimum length of a $(u,v)$-path. 
We remind that $G$ is \emph{connected} if for every two vertices $u$ and $v$, $G$ has a $(u,v)$-path. We always assume that the considered graphs are connected if it is not explicitly said to be otherwise. If $G$ is disconnected, then the inclusion-wise maximal induced connected subgraphs of $G$ are \emph{(connected) components}. 
An edge $e=uv$ is a \emph{bridge} of a connected graph $G$ if $G-e$  is disconnected; note that $G-e$ has exactly two connected components and one of them contains $u$ and the other $v$. We also remind that $T$ is a \emph{spanning tree} of $G$ if $T$ is a tree subgraph of $G$ with the same set of vertices as $G$. 

A \emph{matching} $M$ in a graph $G$ is a set of edges with distinct endpoints. A vertex $v$ is \emph{saturated} in a matching $M$ if $v$ is incident to an edge of $M$.

A set of vertices $S$ of a graph $G$ is a \emph{vertex cover} if each edge of $G$ has at least one of its endpoints in $G$. The \emph{vertex cover number} of $G$ is the minimum size of a vertex cover. Note that for a vertex cover $S$, the set $I=V(G)\setminus S$ is an \emph{independent set}, that is, any two distinct vertices of $I$ are not adjacent. 
It is well-know (see~\cite{CyganFKLMPPS15}) that it is \classFPT to decide whether $G$ has a vertex cover of size at most $k$ when the problem is parameterized by $k$. The currently best algorithm, given by Chen, Kanj, and Xia~\cite{ChenKX10},  runs in $1.2738^k\cdot n^{\Oh(1)}$ time. 

A set of edges $S$ of a graph $G$ is a \emph{feedback edge set} if $G-S$ has no cycle. The \emph{cyclomatic number} of a graph $G$ is the minimum size of a feedback edge set. It is well-known, that for a connected graph $G$, the cyclomatic number is $m-n+1$ and a feedback edge set can be found in linear time by constructing a spanning tree (see, e.g.,~\cite{CormenLRS09,Diestel12}).

\paragraph{Broadcasting.} 
Let $G$ be a graph and let $s\in V(G)$ be a source vertex from which a message is broadcasted. 
In general, a broadcasting protocol is a mapping that for each round $i\geq 1$, assigns to each vertex $v\in V(G)$ that is either a source or has received the message in rounds $1,\ldots,i-1$, a neighbor $u$ to which $v$ sends the message in the $i$-th round.  However, it is convenient to note that it can be assumed that each vertex $v$ that got the message, in the next $d\leq d_G(v)$ rounds, transmits the message to some neighbors in a certain order in such a way that each vertex receives the message  only once. 
This allows us to formally define a \emph{broadcasting protocol} as a pair $(T,\{C(v)\mid v\in V(T)\})$, where 
$T$ is a spanning tree of $G$ rooted in $s$ and for each $v\in V(T)$, $C(v)$ is an ordered set of children of $v$ in $T$. As soon as $v$ gets the message, $v$ starts to send it to the children in $T$ in the order defined by $C(v)$.
For $G$ and  $s\in V(G)$, we use $b(G,s)$  to denote the minimum integer $t\geq 0$ such that there is a broadcasting protocol such that every vertex of $G$ gets the message after $t$ rounds. 
We say that a broadcasting protocol ensuring that every vertex gets a message in  $b(s,G)$ rounds is  \emph{optimal}. 
We use the following straightforward observation.

\begin{observation}\label{obs:span}
Let $T$ be a spanning tree of $G$ and let $s\in V(G)$. Then $b(G,s)\leq b(T,s)$.
\end{observation}

Observation~\ref{obs:span} can be generalized as follows.

\begin{observation}\label{obs:sub}
Let $T$ be a tree subgraph of $G$ and let $s\in V(T)$. Then $b(G,s)\leq b(T,s)+|V(G)\setminus V(T)|$.
\end{observation}

\begin{proof}
Let $P=(T,\{C(v)\mid v\in V(T)\})$ an optimal broadcasting protocol for $T$. We extend it to a protocol for $G$ as follows. First, we construct a spanning subtree $T'$ of $G$ by extending $T$, that is, $T$ is a subtree of $T'$. For every $v\in V(T)$, we construct the ordered set of the children $C'(v)$ by appending to the end of $C(v)$ the children of $v$ in $T'$ that are not in $T$ in arbitrary order. For every $v\in V(T')\setminus V(T)$, we define $C'(v)$ to be an arbitrary ordered set of the children of $v$. We have that $(T',\{C'(v)\mid v\in V(T)\})$ is a broadcasting protocol for $G$  that takes at most $b(T,s)+|V(G)\setminus V(T)|$ rounds to transfer the message to the vertices of $G$, because the vertices of $T$ get the message in the first $b(T,s)$ rounds and in every subsequent round, at least one vertex of $V(G)\setminus V(T)$ receives the message, unless all the vertices are already aware of the message.
\end{proof}

As it was proved by Proskurowski~\cite{Proskurowski81} and Slater, Cockayne, and Hedetniemi~\cite{SlaterCH81}, $b(G,s)$ can be computed in linear time for trees by dynamic programming.

\begin{lemma}[\cite{Proskurowski81,SlaterCH81}]\label{lem:trees}
For an $n$-vertex tree $T$ and $s\in V(T)$, $b(T,s)$ can be computed in $\Oh(n)$ time. 
\end{lemma}

\section{Exact Algorithm for the Broadcast Problem}

We prove that \probTB can be solved in a single-exponential time by an exact algorithm, hence establishing Theorem~\ref{thm:exact}. We do so using dynamic programming over subsets (see the textbook~\cite{FominK10} for an introduction to exact exponential algorithms). 

\paragraph{Proof of Theorem~\ref{thm:exact}.}

Let $(G,s,t)$ be an instance of \probTB. For each $i\in\{0,\ldots,t\}$, we enumerate all subsets of  vertices $X\subseteq V(G)$ containing $s$ such that $G[X]$ is a connected graph and $b(G[X],s)\leq i$. 
We denote these families of sets  $\mathcal{L}_i$ for $i\in\{0,\ldots,t\}$. Observe that $(G,s,t)$ is a yes-instance if and only if $V(G)\in \mathcal{L}_t$.
For $i=0$, it is straightforward that $\mathcal{L}_i$ contains the unique set $\{s\}$. 
We consecutively compute $\mathcal{L}_i$ for $i=1,2,\ldots,t$ using the following claim.

\begin{claim}\label{cl:match}
Let $X\subseteq V(G)$ such that $s\in X$ and $G[X]$ is connected, and let $i\geq 1$. Then $b(G[X],s)\leq i$ if and only if there is $Y\subseteq X$ such that (i) $s\in Y$, $G[Y]$ is connected, $b(G[Y],s)\leq i-1$, and (ii) either $X=Y$ or the bipartite graph $H$ with the set of vertices $X$ and the edges $E(H)=\{uv\in E(G)\mid u\in Y,~v\in X\setminus Y\}$ has a matching saturating every vertex of $X\setminus Y$. 
\end{claim}

\begin{proof}[Proof of Claim~\ref{cl:match}]
Suppose that for $X\subseteq V(G)$, it holds that  $s\in X$, $G[X]$ is connected, and $b(G[X],s)\leq i$. Consider an optimal broadcasting protocol for $G[X]$ and $s$. Let $Y\subseteq X$ be the set of vertices that receive the message in the first $i-1$ rounds. We have that (i) holds for $Y$. Condition (ii) is trivial if $X=Y$. Assume that this is not the case and 
$X\setminus Y\neq\emptyset$. Then for every $v\in X\setminus Y$ there is $u_v\in Y$ such that $v$ receives the message of the $i$-th round from $u_v$. We obtain that $M=\{vu_v\mid v\in X\setminus Y\}$ is a matching in $H$ saturating every vertex of $X\setminus Y$.  

Assume now that there is $Y\subseteq X$ satisfying (i) and (ii). Then there is a broadcasting protocol for $G[Y]$ and $s$ such that every vertex of $Y$ gets the message in $i-1$ rounds. If $X=Y$, then $b(G[X],s)\leq i-1$. Suppose that $X\setminus Y\neq\emptyset$. Then $H$ has a matching $M$ which saturates every vertex of $X\in Y$. We extend the protocol for $Y$ and $s$ by defining that each vertex $u\in Y$ which is saturated in $M$ sends the message to the neighbor along the edge of $M$ incident to $v$ in the last available round. Thus each vertex of $X\setminus Y$  gets the message in $i$  rounds.
\end{proof}

Suppose that $i\geq 1$ and the family $\mathcal{L}_{i-1}$ is already constructed. Initially, we set $\mathcal{L}_i:=\mathcal{L}_{i-1}$. 
Then  we consider all subsets $X\subseteq V(G)$ and proper $Y\subset X$ such that $s\in X$,
$G[X]$ is connected, and $Y\in\mathcal{L}_{i-1}$. For each pair $(X,Y)$ of such sets, we check whether  the bipartite graph $H$ with the set of vertices $X$ and the edges $E(H)=\{uv\in E(G)\mid u\in Y,~v\in X\setminus Y\}$ has a matching saturating every vertex of $X\setminus Y$, and if this holds, then we set $\mathcal{L}_i:=\mathcal{L}_i\cup \{X\}$. 

Claim~\ref{cl:match} guarantees  correctness of the construction of $\mathcal{L}_i$ from $\mathcal{L}_{i-1}$. To evaluate the running time, note that we consider at most $3^n$ pairs of sets $(X,Y)$. Then for each pair, $H$ can be constructed in polynomial time and then the existence of a matching saturating the vertices of $X\setminus Y$ can be verified in polynomial time in the standard way (see, e.g.,~\cite{CormenLRS09}). Thus $\mathcal{L}_i$ can be constructed in $3^n\cdot n^{\Oh(1)}$ time. Taking into account that we iterate for $i=1,2,\ldots,t$ and we can assume without loss of generality that $t\leq n-1$, we have that the total running time is $3^n\cdot n^{\Oh(1)}$.
\qed

 \section{\probTB parameterized by the cyclomatic number}\label{sec:cyc} 
 In this section, we prove Theorem~\ref{thm:cyc}. We need some auxiliary results. 
 
 Let $T$ be a tree and let $x$ and $y$ be distinct leaves, that is, vertices of degree one in $T$. For an integer $h\geq 0$, we use $b_h(T,x,y)$ to denote the minimum number of rounds needed to broadcast the message from the source $x$ to $y$ in such a way that every vertex of $T$ gets the message in at most $h$ rounds. We assume that $b_h(T,x,y)=+\infty$ if $b(T,x)>h$. We prove that 
 $b_h(T,x,y)$ can be computed in linear time  similarly to $b(T,s)$ (see~\cite{Proskurowski81,SlaterCH81}).  The difference is that it is more convenient to use recursion instead of dynamic programming. 
 
 \begin{lemma}\label{lem:bxy}
 For an $n$-vertex tree $T$ with given  distinct leaves $x$ and $y$ of $T$ and an integer $h\geq 0$, $b_h(T,x,y)$ can be computed in $\Oh(n)$ time. 
  \end{lemma}
 
\begin{proof}
Let $T$ be a tree and let $x$ and $y$ be leaves of $T$. Let also $h\geq 0$ be an integer. We use Lemma~\ref{lem:trees} to check whether $b(T,x)\leq h$ and set $b_h(T,x,y)=+\infty$ if 
$b(T,x)> h$, because in this case the message cannot be transmitted to every vertex of $T$ in $h$ rounds. From now we assume that $b(T,x)\leq h$. Notice that this implies that $b_h(T,x,y)<+\infty$.

Let $P=v_1\ldots v_\ell$ be the $(x,y)$-path in $T$, $x=v_1$ and $y=v_\ell$. Note that if $\ell=2$, then, trivially, $b_h(T,x,y)=1$. Suppose that $\ell\geq 3$.  
Denote by $T'$ the subtree of $T$ containing $v_2$ and $y$ such that $v_2$ is a leaf of $T'$. We construct recurrences that allow to compute $b_h(T,x,y)$ from $b_{h'}(T',v_2,y)$ for a certain $h'<h$.

If $d_T(v_2)=2$, then $b_h(T,x,y)=1+b_{h-1}(T',v_2,y)$ by definition.

Suppose that $d_T(v_2)\geq 3$. Denote by $u_1,\ldots,u_k$ the neighbors of $v_2$ distinct from $v_1$ and $v_3$, and let $T_1,\ldots,T_k$ be the connected components of $T-v_2$ containing $u_1,\ldots,u_k$, respectively. Notice that $k\leq h-2$, because at least $k+2$ rounds are needed to transmit the message to $u_1,\ldots,u_k$ and $v_3$. Observe also that the message is broadcasted from $v_2$ to every $u_i$ and to broadcast the message from each $u_i$ to other vertices of $T_i$, we can use an optimal protocol for $T_i$ with the source $u_i$.  
We compute $b(T_i,u_i)$ for $i\in\{1,\ldots,k\}$ and assume that $b(T_1,u_1)\geq \dots\geq b(T_k,u_k)$.  
Observe that $\max\{b(T_i,u_i)+i\mid i\in\{1,\ldots,k\}\}\leq  h-1$, because at least $\max\{b(T_i,u_i)+i\mid i\in\{1,\ldots,k\}\}+1$ rounds are necessary to transmit the message to all the vertices of $T_1,\ldots,T_k$. 

 Because  $\max\{b(T_i,u_i)+i\mid i\in\{1,\ldots,k\}\}\leq h-1$, there is $j\in\{1,\ldots,k+1\}$ such that
 either $j\leq k$ and $\max\{b(T_i,u_i)+i+1\mid i\in\{j,\ldots,k\}\}\leq h-1$ or $j=k+1$. We find the minimum value of $j$ satisfying this condition. Observe that $j+1$ is the minimum number of a round such that the message can be transmitted from $v_2$ to $v_3$ in such a way that every vertex of $T_1,\ldots,T_k$ gets the message in at most $h$ rounds. Because we aim  to minimize the transmission time to $y=v_k$, we set 
 $b_h(T,x,y)=j+b_{h-j}(T',v_2,y)$. 
 
 Correctness of the algorithm follows from its description. To evaluate the running time, note that to compute $b_h(T,x,y)$, we have to find $P=v_1\ldots v_\ell$ and then for each $i\in\{2,\ldots,\ell\}$, compute the neighbors $u_1^i,\ldots,u_{k_i}^i$ of $v_i$ distinct from $v_{i-1}$ and $v_{i+1}$ and the components $T_1^i,\ldots,T_{k_i}^i$ of $T-V(P)$ containing  $u_1^i,\ldots,u_{k_i}^i$. All these computations can be done in the total linear time.  Also in $\Oh(n)$ time, we can order  $T_1^i,\ldots,T_{k_i}^i$  by the decrease $b(T_p^i,u_p)$ for $p\in \{1,\ldots,k_i\}$. Then computing $b_h(T,x,y)$ can be done in linear time following the recurrences. 
\end{proof}
 
We also need a subroutine computing the minimum number of rounds for broadcasting from two sources with the additional constraint that the second source starts sending the message with a delay. 
 Let $T$ be a tree and let $x$ and $y$ be distinct leaves of $T$.  Let also $h\geq 0$ be an integer. We use $d_h(T,x,y)$ to denote the minimum rounds needed to broadcast the message from $x$ and $y$ to every vertex of $T-y$ in such a way that $y$ can send the message starting from the $(h+1)$-th round (that is, we assume that $y$ gets the message from outside in the $h$-th round).   
 
   \begin{lemma}\label{lem:dxy}
 For an $n$-vertex tree $T$ with given  distinct leaves $x$ and $y$ of $T$ and an integer $h\geq 0$, $d_h(T,x,y)$ can be computed in $\Oh(n^2)$ time. 
  \end{lemma} 
 
 \begin{proof}
 Let $T$ be a tree and let $x$ and $y$ be leaves of $T$. Let also $h\geq 0$ be an integer. Consider the $(x,y)$-path $P=v_1\ldots v_\ell$ in $T$, where $x=v_1$ and $y=v_\ell$. 
 Then we have that 
 $d_h(T,x,y)=b(T-y,x)$ if $b(T-y,x)\leq h$.
 If $b(T-y,x)>h$, then
 \begin{equation}\label{eq:split}
 d_h(T,x,y)=\min_{i\in\{2,\ldots,\ell\}}\max\{b(T_1^i,x),b(T_2^i,y)+h\},
  \end{equation}
 where  $T_1^i$ is the connected component of $T-v_{i-1}v_i$ containing $x$ and $T_2^i$ is the connected component of $T-v_{i-1}v_i$ containing $y$.
 The equality (\ref{eq:split}) follows from the observation that the sets of vertices receiving the message from $x$ and $y$, respectively, should form a partition of $V(T)$, where each part induces a subtree.  The path $P$ can be found in linear time. Then computing $d_h(T,x,y)$ can be done in $\Oh(n^2)$ time using Lemma~\ref{lem:trees}.  
 \end{proof}

 Now we are ready to prove Theorem~\ref{thm:cyc}.
 
 \subsection*{Proof of Theorem~\ref{thm:cyc}}
 Let $(G,s,t)$ be an instance of \probTB. If $G$ is a tree, then we can compute $b(G,s)$ in linear time using Lemma~\ref{lem:trees}. Assume that this is not the case, and let $k=m-n+1\geq 1$ be the cyclomatic number of $G$. We find in linear time a feedback edge set $S$ of size $k$ by finding an arbitrary spanning tree $F$ of $G$ and setting $S=E(G)\setminus E(F)$. 
 
We iteratively construct the set $U$ as follows. Initially, we set $U:=W=\{s\}\cup\{v\in V(G)\mid v\text{ is an endpoint of an edge of }S\}$. Then while $G-U$ has a vertex $v$ such that $G$ has three internally vertex-disjoint  paths joining $v$ and $U$, we set $U:=U\cup\{v\}$. The properties of $U$ are summarized in the following claims.

\begin{claim}\label{cl:U-size}
$|U|\leq 4k$.
\end{claim}

\begin{proof}[Proof of Claim~\ref{cl:U-size}]
The claim is trivial if $U=W$, because $|U|\leq 2k+1\leq 4k$ in this case. Assume that $|U|>|W|$, that is, $G-W$ has a  vertex $v$ such that $G$ has three internally vertex-disjoint  paths joining $v$ and $W$. Consider the graph $G'$ obtained from $G$ by the iterative deletion of all the vertices of degree one that are not in $W$. Observe that $F=G'-S$ is a tree whose leaves are in $W$, and each vertex $v\in V(G')\setminus W$ can be joined by three vertex-disjoint paths with $W$ if and only if $v$ is a vertex of degree at least three in $F$.  It is folklore knowledge that if $\ell$ is the number of leaves and $p$ is the number of vertices of degree at least three in a tree, then $p\leq \ell-2$. This implies that $U$ is constructed from $W$ by adding at most $|W|-2\leq 2k-1$ vertices. Hence, 
$|U|\leq 4k$.
\end{proof}

\begin{claim}\label{cl:U-comp}
For each connected component $F$ of $G-U$, $F$ is a tree such that each vertex $x\in U$ has at most one neighbor in $F$ and
\begin{itemize}
\item[(i)] either $U$ has a unique vertex $x$ that has a neighbor in $F$,
\item[(ii)] or $U$ contain exactly two vertices $x$ and $y$ having neighbors in $F$. 
\end{itemize}
\end{claim}

\begin{proof}[Proof of Claim~\ref{cl:U-comp}]
We have that each connected component $F$ of $G-U$ is a tree because $G-S$ is a tree and the endpoints of the edges of $S$ are in $U$. By the same reasons, each vertex $x\in U$ is adjacent to at most one vertex of each $F$. Consider arbitrary $F$. If only one vertex $x\in U$ has a neighbor in $F$, then (i) is fulfilled.  Suppose that at least two vertices of $U$ are adjacent to some vertices of $F$. Then  exactly two vertices $x,y\in U$ have neighbors in $F$, because otherwise $F$ would have a vertex $v$ which could be joined with $U$ by three vertex-disjoint paths but such a vertex would be included in $U$. This implies that (ii) holds and concludes the proof.  
\end{proof}
 
If $F$ is a connected component of $G-U$ satisfying (i) of Claim~\ref{cl:U-comp}, then we say that $F$ is a $x$-tree and $x$ is its \emph{anchor}.  
For a connected component $F$ of $G-U$ satisfying (ii), we say that $F$ is an \emph{$(x,y)$-tree} and call $x$ and $y$ \emph{anchors} of $F$.   We also say that $F$ is \emph{anchored} in $x$ ($x$ and $y$, respectively). Because $G-S$ is a tree and $|U|\leq 4k-1$ by Claim~\ref{cl:U-size}, we immediately obtain the next property.

\begin{claim}\label{cl:xy-trees}
For every distinct $x,y\in U$, $G-U$ has at most one $(x,y)$-tree. Furthermore, the graph $H$ with $V(H)=U$ such that $xy\in E(H)$ if and only if $G-U$ has an $(x,y)$-tree is a forest. In particular, 
the total number of $(x,y)$-trees is at most $4k-1$. 
\end{claim} 
 
 To prove the theorem, we have to verify the existence of a broadcasting protocol $P=(T,\{C(v)\mid v\in V(T)\})$ that ensures that every vertex receives the message after at most $t$ rounds. To do it, we guess the \emph{scheme} of $P$ restricted to $U$. Namely, we consider the graph $G'$ obtained from $G$ by the deletion of the vertices of $x$-trees for all $x\in U$ and for each vertex $x\in U$, we guess how the message is broadcasted to $x$ and from $x$ to the neighbors of $x$ in $G'$. Notice that $T'=T[V(G')]$ is a tree by the definition of $G'$. Observe also that for each $x$-tree $F$ for $x\in U$, the message is broadcasted  to the vertices of $F$ from $x$, because $s\in U$. In particular, this means that the parents of the vertices of $U$ in $T$ are in $G'$. 
 For each $v\in U$ distinct from $s$, we guess its parent $p(v)\in V(G')$ in $T'$ and assume that $p(s)=s$. Then for each $v\in V(G)$, we guess the ordered subset $R(v)$ of vertices  of $N_{G'}(v)\setminus\{p(v)\}$ such that 
 $R(v)=C(v)\cap N_{G'}(v)$.  We guess $p(v)$ and $R(v)$ for $v\in U$ by considering all possible choices.  To guess $R(v)$ for each $v\in U$, we first guess the (unordered) set $S(v)$ and then consider all possible orderings of the elements of $S(v)$.
  The selection of $p(v)$ and $S(v)$ is done by brute force. However,  we are only interested in choices, where the selection of the neighbors $p(v)$ and $S(v)$ of $v$ for $v\in U$ can be extended to a spanning tree $T'$ of $G'$.  
 
 Let $T'$ be  an arbitrary spanning tree of $G'$ rooted in $s$. Let $T''$ be the tree obtained from $T'$ by the iterative deletion of leaves not included in $U$. Observe that $T''$ is a tree such that $U\subseteq V(T'')$ and each leaf of $T''$ is a vertex of $U$. By Claim~\ref{cl:U-comp}, each edge of $T''$ is either an edge of $G[U]$ or is an edge of an $(x,y)$-path $Q$ for distinct $x,y\in U$ such that the internal vertices of $Q$ are the vertices of the $(x,y)$-tree $F$; in the second case, each edge of $Q$ is in $T''$. Notice also that $s\in V(T'')$ and for each $v\in U$ distinct from $s$, the parent of $v$ in $T$ is the parent of $v$ in $T''$ with respect to the source vertex $s$. Hence, our first step in constructing $p(v)$ and $S(v)$, is to consider all possible choices of $T''$. Observe that $G[U]$ has at most $\binom{4k}{2}$ edges by  Claim~\ref{cl:U-size} and the total number of $(x,y)$-trees is at most $4k-1$ by Claim~\ref{cl:xy-trees}. Because $T''$ is a tree, it contains $|U|-1$ edges of $G[U]$ and $(x,y)$-paths $Q$ in total. We obtain that we have $k^{\Oh(k)}$  possibilities to choose $T''$. From now, we assume that $T''$ is fixed. 
 
 The choice of $T''$ defines $p(v)$ for $v\in U\setminus \{s\}$. For each $v\in U$, we initiate the construction of $S(v)$ by including in the set the neighbors of $v$ in $T''$ distinct from $p(v)$. We proceed with  guessing of $S(v)$ by considering $(x,y)$-trees $F$ for $x,y\in U$ such that the $(x,y)$-path with the internal vertices in $F$ is not included in $T''$. Clearly, the vertices of every $F$ of such a type should receive the message either via $x$, or via $y$, or via both $x$ and $y$. Let $F$ be an $(x,y)$-tree of this type. Denote by $x'$ and $y'$ the neighbors of $x$ and $y$, respectively. We have that either $x'\in S(x)$ and $y'\notin S(y)$, or $x\notin S(x)$ and $y\in S(Y)$, or  $x'\in S(x)$, $y'\in S(y)$ and $x'\neq y'$. Thus, we have three choices for $F$. By Claim~\ref{cl:xy-trees}, the total number of choices is $2^{\Oh(k)}$. We go over all the choices and include the vertices to the sets $S(v)$ for $v\in U$ with respect to them. By Claim~\ref{cl:U-comp}, this concludes the construction of the sets $S(v)$. From now, we assume that $S(v)$ for $v\in U$ are fixed. 
 
 We construct the ordered sets $R(v)$ by considering all possible orderings of the elements of $S(v)$. The number of these orderings is $\Pi_{v\in U}(|S(v)|!)$. Recall that the sets $S(v)$ are sets of neighbors of $v$ in a spanning tree of $G'$. This and Claims~\ref{cl:U-comp} and \ref{cl:xy-trees} imply that $\sum_{v\in U}|S(v)|\leq 2(|U|-1)+2(4k-1)\leq 16k$. Therefore, the total number of orderings is
   $\Pi_{v\in U}(|S(v)|!)=k^{\Oh(k)}$. This completes the construction of $R(v)$. Now we can assume that $p(v)$ for each $v\in U\setminus \{s\}$ and $R(v)$ for each $v\in U$ are given.
     
 The final part of our algorithm is checking whether the  guessed scheme for a broadcasting protocol can be extended to the protocol itself.  This is done in two stages.
 
 In the first stage, we compute for each $v\in U$, the minimum number $r(v)$ of a round in which $v$ can receive the message and the ordered set $C(v)$.  Initially, we set $r(s)=0$ and set $X:=\{s\}$. Then we iteratively either compute $C(v)$ for $v\in X$ or extend $X$ by including a new vertex $v\in U\setminus X$ and computing $r(v)$. We proceed until we get $X=U$ and compute $C(v)$ for     
  every $v\in U$. We also stop and discard the current choice of the scheme if we conclude that the choice cannot be extended to a broadcasting protocol terminating in at most $t$ steps.
  
 Suppose that there is $v\in X$ such that $C(v)$ is not constructed yet. Notice that $r(v)$ is already computed. To construct $C(v)$, we observe that for each $v$-tree $F$ anchored in $v$, the vertices of $F$ should receive the message via $v$. Hence, to construct $C(v)$, we extend $R(v)$ by inserting the neighbors of $v$ in the $v$-trees. 
 If there is no $v$-tree anchored in $v$, then we simply set $C(v)=R(v)$. Assume that this is not the case and let 
 $F_1,\ldots,F_k$ be the $v$-trees anchored in $v$. Denote by $u_1,\ldots,u_k$ the neighbors of $v$ in $T_1,\ldots,T_k$, respectively. Because the message is broadcasted from $u$ to each $u_i$, we can assume that to broadcast the message from $u_i$ to the other vertices of $T_i$, an optimal protocol requiring $b(T_i,u_i)$ rounds is used. We compute the values $b(T_i,u_i)$ for all $i\in\{1,\ldots,k\}$ and assume that $b(T_1,u_1)\geq\dots\geq b(T_k,u_k)$. 
 
 If $r(v)+|R(v)|+k>t$, we discard the current choice of the scheme, because we cannot transmit the message to the neighbors of $v$ in $t$ rounds.  Notice also that 
 $r(v)+\max\{b(T_i,u_i)+i\mid i\in\{1,\ldots,k\}\}$  rounds are needed to transmit the message to the vertices of all $v$-trees. Hence, if  
 $r(v)+\max\{b(T_i,u_i)+i\mid i\in \{1,\ldots,k\}\}>t$, we discard the considered scheme. From now on, we assume that $r(v)+|R(v)|+k\leq t$ and 
 $r(v)+\max\{b(T_i,u_i)+i\mid i\{1\in,\ldots,k\}\}\leq t$.
 
 The main idea for constructing $C(v)$ is to ensure that the message is sent to the vertices of $R(v)$ as early as possible. To achieve this, we put $u_1,\ldots,u_k$ in $C(v)$ in such a way, that the message is sent to each $u_i$ as late as possible. 
 Since $|C(v)|=|R(v)|+k$, we represent $C(v)$ as an $|R(v)|+k$-element array whose elements are indexed $1,2,\ldots,|R(v)|+k$. 
 Because  $b(T_1,u_1)\geq\dots\geq b(T_k,u_k)$, we can assume that the ordering of the vertices $u_1,\ldots,u_k$ in $C(v)$ is $(u_1,\ldots,u_k)$. Therefore, we insert $u_i$ in $C(v)$ consecutively for $i=k,k-1,\ldots,1$.  
 
 Suppose that $i\in\{1,\ldots,k\}$ and $u_{i+1},\ldots,u_k$ are in $C(v)$. Denote by $h_{i+1}$ the index of $u_{i+1}$ assuming that $h_{k+1}=|R(v)|+k+1$. We find maximum positive integer $h<h_{i+1}$ such that $r(v)+h+b(T_i,u_i)\leq t$ and set the index $h_i=h$ for $u_i$. In words, we find the maximum index that is prior to the index of $u_{i+1}$ such that if we transmit the message from $v$ to $u_i$ in the $h$-th round after $v$ got aware of the message, then the vertices of $T_i$ still may get the message in $t$ rounds. 
 After placing $u_1,\ldots,u_k$ into the array, we place the vertices of $R(v)$ in the remaining $|R(v)|$ places following the order in $R(v)$. This  completes the construction of $C(v)$.
 
 Suppose that $U\setminus X\neq\emptyset$ and for each $v\in X$, $C(v)$ is given. We assume that for each $v\in X$, the elements of $C(v)$ are indexed $1,
 \ldots,|C(v)|$ according to the order.  By the constriction of the schemes, there is $y\in U\setminus X$ such that $y$ receives the message from some vertex $x\in X$ either directly or via some $(x,y)$-tree $F$ anchored in $x$ and $y$. We find such a vertex $y$, compute $r(y)$, and include $y$ in $X$. 
 
 If there is $y\in U\setminus X$ such that $p(y)=x\in X$, then we set $r(y)=r(x)+h$, where $h$ is the index of $y$ in $C(v)$ and set $X:=X\cup\{y\}$. Since $r(x)$ is the minimum number of a round when $x$ gets the message, $r(y)$ is the minimum number of a round in which $y$ gets the message. Suppose that such a vertex $y$ does not exist. Then by the construction of the considered scheme, there are $x\in X$ and $y\in U\setminus X$ such that the tree $T''$ which was used to construct the scheme contains an $(x,y)$-path whose internal vertices are in the $(x,y)$-tree $F$ anchored in $x$ and $y$. This means that the neighbor $x'$ of $x$ in $F$ is included in $C(v)$. Let $h$ be the index of $x'$ in $C(v)$. We also have that $y'=p(y)$ is the unique neighbor of $y$ in $F$. In other words, we have to transmit the message from $x$ to $y$ according to the scheme.  To compute $r(y)$, we have to transmit the message as fast as possible. For this, we use Lemma~\ref{lem:bxy}. Notice that the vertices of $F$ should receive the message in at most $t'=t-r(x)-h+1$ rounds because $x$ receives the message in the round $r(v)$ and $h-1$ vertices of $C(v)$ get the message before $x'$. 
 Let $F'$ be the tree obtained from $F$ by adding the vertices $x,y$ and the edges $xx',yy'$. We compute $b_{t'}(F',x,y)$ using the algorithm from  Lemma~\ref{lem:bxy}.  
 If $b_{t'}(F',x,y)=+\infty$, we discard the considered scheme because $y$ cannot receive the message in $t$ rounds. Otherwise, we set $r(y)=r(x)+(h-1)+b_{t'}(F',x,y)$ and set $X:=X\cup\{y\}$.
 
 This completes the first stage where we compute $r(v)$ and $C(v)$ for $v\in U$. Observe that if we completed this stage without discarding the considered choice of the scheme, we already have a partially constructed broadcasting protocol that ensures that (i) the vertices of $v$-trees for $v\in U$ get the message in at most $t$ rounds, (ii) the vertices of $(x,y)$-trees that are assigned by the scheme to transmit  the message from $x$ to $y$ receive the message in at most $t$ rounds, and (iii) for each $v\in U$, $r(v)$ is the minimum number of a round when $v$ can receive the message according to the scheme. By Claim~\ref{cl:U-comp}, it remains to check whether the vertices of $(x,y)$-trees $F$ that are not assigned by the scheme to transmit the message from $x$ to $y$ or vice versa can receive the message in at most $t$ rounds. We do it using Lemmas~\ref{lem:trees} and~\ref{lem:dxy}.
 
 Suppose that $F$ is a $(x,y)$-tree anchored in $x,y\in U$ such that $p(x),p(y)\notin V(F)$, that is, $F$ is not assigned by the scheme to transmit the message from $x$ to $y$ or vice versa. 
 Let $x'$ and $y'$ be the neighbors in $F$ of $x$ and $y$, respectively. By the construction of the schemes, we have three cases: (i) $x'\in C(x)$ and $y'\notin C(y)$, (ii) $x'\notin C(x)$, $y'\in C(y)$, 
 and (iii) $x'\in C(x)$, $y'\in C(y)$, and $x'\neq y'$.
 
 In case (i), the vertices of $F$ should receive the message via $x$. Clearly, we can use an optimal protocol for $F$ with the source $x'$ to transmit the message from $x'$. 
  Let $h$ be the index of $x'$ in $C(x)$. We use Lemma~\ref{lem:trees}, to verify whether $t-r(x)-h\geq b(F,x)$. If the inequality holds, we conclude that the message can be transmitted to the vertices of $F$ in at most $t$ rounds. Otherwise, we conclude that this is impossible and discard the scheme. Case (ii) is symmetric and the arguments are the same.
  
  Suppose that $x'\in C(x)$, $y'\in C(y)$, and $x'\neq y'$. Then the vertices of $F$ are receiving the message from both $x$ and $y$. Denote by $i$ and $j$ the indexes of $x'$ and $y'$ in $C(x)$ and $C(y)$, respectively. By symmetry, we assume without loss of generality that $r(x)+i\leq r(y)+j$ and let $h=(r(y)+j)-(r(x)+i)$. Notice that the vertices of $F$ start to get the message after the round $r(v)+i-1$. Denote by $F'$ the tree obtained from $F$ by adding the vertices $x,y$ and the edges $xx',yy'$.
  We use Lemma~\ref{lem:dxy} and compute $d_h(F',x,y)$. If  $d_h(F',x,y)\leq t-r(v)-i+1$, then we obtain that the message can be transmitted to the vertices of $F$ in at most $t$ rounds. Otherwise, we cannot do it and discard the scheme.   
  
 This completes the description of the second stage of the verification of whether the considered scheme can be extended to a broadcasting protocol terminating in at most $t$ rounds.
 
 If we find a scheme that allows us to conclude that the message can be broadcasted in at most $t$ rounds, we conclude that $(G,s,t)$ is a yes-instance. Otherwise, if every scheme gets discarded, we return that $(G,s,t)$ is a no-instance of \probTB. This concludes the description of the algorithm.
 x
 To argue that the algorithm is correct, note that if there is a scheme for which we report that $(G,s,t)$ is a yes-instance, then the description of the verification procedure of the extendability of the scheme implies the existence of a broadcasting protocol terminating in at most $t$ rounds. For the opposite direction, assume that $(G,s,t)$ is a yes-instance of \probTB. Then there is a broadcasting protocol $P$ terminating in at most $t$ rounds. For $v\in V(G)$, denote by $r_P(v)$ the number of the round on which $v$ gets the message. Among all  $P=(T,\{C'(x)\mid x\in V(T)\})$ terminating in $t$ rounds, we select $P$ to be protocol such that $\sum_{v\in U}r_P(v)$ is minimum. We construct $G'$ as described in the algorithm and consider $T'=T[V(G')]$. Then we construct $T''$ from $T'$ by the iterative deletion of leaves that are not included in $U$. Then we define the scheme by setting $p(v)$ to be parent of each $v\in U$ in $T''$ with respect to the root $s$ ($p(s)=s$) and defining $R(v)=C'(v)\cap N_{T''}(v)$ for $v\in U$. Observe that the obtained scheme is one of the schemes constructed by the algorithm. Then the description of the algorithm implies that the verification step of the algorithm returns the yes-answer for the scheme. In particular, we get that for every $v\in U$, $r(v)=r_P(v)$ and the indices of the elements of  $R(v)$ in the set $C(v)$ obtained by the algorithm are the same as in $C'(v)$. 
 
To evaluate the running time, observe that the feedback edge set $S$ and the set of vertices $U$ are constructed in polynomial time. Given $U$, we can list all $x$-trees and $(x,y)$-trees for $x,y\in U$   
 in polynomial time. Then the schemes can be constructed in $k^{\Oh(k)}$ time. For each scheme, the verification step requires polynomial time as we use the polynomial algorithms from Lemmas~\ref{lem:trees}, \ref{lem:bxy}, and~\ref{lem:dxy}. We conclude that the total running time is $2^{\Oh(k\log k)}\cdot n^{\Oh(1)}$. This completes the proof.  
  
\section{\probTB parameterized by the vertex cover number}\label{sec:vc}
In this section, we prove Theorem~\ref{thm:vc}. Recall that we aim to show that \probTB is \classFPT on graphs with the vertex cover number at most $k$ when the problem is parameterized by $k$. We start with some auxiliary claims about the broadcasting on a graph with a given vertex cover $S$.

\begin{lemma}\label{lem:vc}
Let $G$ be a graph with at least one edge and $s\in V(G)$. Let also $S$ be a vertex cover of $G$. Then there is an optimal broadcasting protocol for $G$ with the source $s$ such that the vertices of $S$ receive the message in at most $2|S|-1$ rounds.
\end{lemma}

\begin{proof}
Let $(T,\{C(v)\mid v\in V(T)\})$ be an optimal broadcasting protocol for $G$ with the source $s$. 
Observe that we can assume without loss of generality that the protocol is \emph{greedy}, meaning that for each round $i$, if it holds that $v$ is a vertex that is aware of the message after $i$ rounds, it would always send the message in the next round to some neighbor unless all the neighbors either got the message in the first $i$ rounds or will receive the message in the $(i+1)$-th round from vertices that are distinct from $v$. In words, a vertex does not stay idle if it can communicate the message to a new recipient. Furthermore, we can assume that 
for each $v\in V(T)$, the children of $v$ of degree one, i.e., leaves of $T$, are at the end of the ordered set $C(v)$. Otherwise, we can rearrange each $C(v)$ to achieve this property and the obtained protocol would be optimal. Also, we can include the leaf children of each $v$ in $C(v)$ in arbitrarily order because any rearrangement of these vertices does not change the number of rounds. Hence, we assume that $(T,\{C(v)\mid v\in V(T)\})$ is a greedy  optimal broadcasting protocol such that
for each $v\in V(T)$, (i) the children of $v$ of degree one are at the end of the ordered set $C(v)$ and (ii) the children of degree one that are in $S$ are before  the children of degree one in $I=V(G)\setminus S$. We claim that every vertex of $S$ gets the message in the first $2|S|-1$ round.

For the sake of contradiction, assume that this is not the case. Then there is a nonnegative integer $h\leq 2|S|-2$ such that neither in round $h$ nor in round $h+1$, there is a vertex of $S$ that receives the message. Let $S'\subseteq S$ be the set of vertices of $S$ that got the message after $h-1$ rounds and let $S''=S\setminus S'$. Let also $I'\subseteq I$ be the set of vertices of $I$ that got the message after $h-1$ rounds and denote by $I''\subseteq I$ the set of vertices of $I$ that get the message in the $h$-th round. Because no vertex of $S$ gets the message in the $h$-th round, we have that $I''\neq\emptyset$. Observe that no vertex of $I'\cup I''$ is adjacent to a vertex of $S''$ as otherwise at least one vertex of $S'$ would get the message in the $(h+1)$-th round. In particular, this means that the vertices of $I''$ are leaves of $T$.
Because $G$ is connected and $T$ is a spanning tree, we have that there is a vertex $v\in S'$ which is adjacent to a vertex $u\in S''$ in $T$. Since the protocol is greedy, $v$ sent the message to a vertex of $I''$ in the $h$-th round. However, this contradicts that the protocol satisfies (i) and (ii) for $v$. If $u$ is not a leaf of $T$, then $v$ should send the message to $u$ instead of a vertex of $I''$ by (i), and if $u$ is a leaf, then $u$ is a leaf in $S$ and $v$ should send the message to $u$ in the $h$-th round by (ii).  The obtained contradiction proves the lemma.
\end{proof}

We also use the bound for the number of vertices of  $I=V(G)\setminus S$ getting the message in the first $p$ rounds.

\begin{lemma}\label{lem:is}
Let $G$ be a graph with at least one edge and $s\in V(G)$. Let also $S$ be a vertex cover of $G$ and $p\geq 1$ be an integer. Then for any  broadcasting protocol for $G$ with the source $s$, 
at most $p|S|$ vertices of $I=V(G)\setminus S$ receive the message in  the first $p$ rounds.
\end{lemma}

\begin{proof}  
Note that each vertex $v\neq s$ in $I$ can receive the message only from its neighbor in $S$. If $s\in S$, then each vertex of $S$ can send the message to at most $p$ neighbors in $p$ round. Therefore, at most $p|S|$ can get the message in $p$ rounds. If $s\in I$, then each vertex of $S$ can send the message to at most $p-1$ neighbors. Hence, at most $(p-1)|S|+1\leq p|S|$ vertices of $I$ can receive the message in $p$ rounds. 
\end{proof}

Now we are ready to prove Theorem~\ref{thm:vc}.

\subsection*{Proof of Theorem~\ref{thm:vc}}
Let $(G,s,t)$ be an instance of \probTB and let $k\geq 0$ be an integer. \probTB is trivial for $k=0$ as $G$ has the empty vertex cover if and only if $G$ has a single vertex (recall that $G$ is connected by our assumption). Also, the problem is trivial if $G$ has no edges. 
Hence, we assume that $G$ has at least two vertices and $k\geq 1$. 

We use the algorithm of Chen, Kanj, and Xia~\cite{ChenKX10}  to find  in $1.2738^k\cdot n^{\Oh(1)}$ time a vertex cover $S$ of $G$ of size at most $k$. If the algorithm fails to find such a set, then we stop and return the answer that $G$ has no vertex cover of size at most $k$. From now on, we assume that $S$ is given and $I=V(G)\setminus S$. If $|I|\leq 2k^2$, then $|V(G)|\leq 2k^2+k$ and we solve the problem in $2^{\Oh(k^2)}\cdot k^{\Oh(1)}$ time using Theorem~\ref{thm:exact}.  So we can assume that $|I|>2k^2$.

We compute the partition of $\{L_1,\ldots,L_p\}$ of $I$ into classes of false twins, that is, two vertices $u$ and $v$ are in the same set $L_i$ if and only if they have exactly the same neighbors in $S$.  Since $|S|\leq k$, we have that $p\leq 2^k$. Observe also that the partition can be constructed in linear time (see~\cite{TedderCHP08}). As it is common for the parameterization by the vertex cover number, we exploit the property that two vertices $u,v\in L_i$ for some $i\in\{1,\ldots,p\}$ that are distinct from $s$ are \emph{indistinguishable}. In particular, given a broadcasting protocol $(T,\{C(v)\mid v\in V(T)\})$, we can exchange $u$ and $v$, that is, replace the former $u$ by $v$ and former $v$ by $u$ in $T$ and the ordered sets,  and this would keep the number of rounds the same.

By Lemma~\ref{lem:is}, at most $2k^2$ vertices of $I$ can receive the message in the first $2k$ rounds. Hence, if $t\leq 2k^2<|I|$, we conclude that $(G,s,t)$ is a no-instance. From now on we can assume that $t>2k^2$. 

By Lemma~\ref{lem:vc}, there is an optimal broadcasting protocol such that all the vertices of $S$ receive the message in the first $2k-1$ rounds. Combining this with Lemma~\ref{lem:is}, we obtain that there is an optimal broadcasting protocol such that the set of vertices $X$ that receive the message in the first $2k$ rounds satisfies the properties (i) $S\subseteq X$ and (ii)
for $Y=S\cap I$, $|Y|\leq 2k^2$. 

We guess $Y$ using the fact that the vertices of the same set $L_i$ are indistinguishable.  We consider all possibilities to choose $p$ nonnegative integers $\ell_1,\ldots,\ell_p$ such that 
(a) $\ell_i\leq |L_i|$ for each $i\in\{1,\ldots,p\}$, (b) if $s\in L_i$ for some $i\in\{1,\ldots,p\}$, then $\ell_i\geq 1$,  and (c) $\ell_1+\dots+\ell_p\leq 2k^2$. Notice that because $p\leq 2^k$, there are 
$2^{\Oh(k^3)}$ possibilities to choose $\ell_1,\ldots,\ell_p$. Given $\ell_1,\ldots,\ell_p$, for each $i\in\{1,\ldots,p\}$, we can select $\ell_i$ vertices from $L_i$ in such a way that if $s\in L_i$, then $s$ is selected, and the other vertices are taken arbitrariy, and then denote the set of selected vertices by $Y_i$. We have that if there is   
 an optimal broadcasting protocol such that the set of vertices $X$ that receive the message in the first $2k$ rounds satisfying (i) and (ii), then there is an optimal broadcasting protocol such that the set of vertices $X$ that receive the message in the first $2k$ rounds is $S\cup(Y_1\cup\dots\cup Y_p)$. We verify whether the choice of $X$ is \emph{feasible}, that is, the vertices of $X$ can get messages in $2k$ rounds, using Theorem~\ref{thm:exact}.
For this we check whether $G[X]$ is connected and if this holds, then run the algorithm from the theorem for $G[X]$, $s$, and $t=2k$. We discard the choice of $X$ if $G[X]$ is disconnected or 
if the algorithm reports that the considered instance is a no-instance.  
Observe that the running time of the algorithm is $2^{\Oh(k^2)}$ in this case. 
 
Notice that if the vertices of $Y_i$ are getting the message in the first $2k$ rounds for some $i\in\{1,\ldots,p\}$, then the remaining $|L_i|-\ell_i$ vertices should receive the message from their neighbors in $S$ in the following $t-2k$ rounds. Notice that because each vertex of $S$ is aware of the message after $2k$ rounds in a hypothetical  solution, any vertex of $S$ can send the message to any neighbor in $I$ and the vertices of $S$ do not send the message to $S$. This allows us to encode a broadcasting protocol for the final steps as a system of linear inequalities over $\mathbb{Z}$. 

For every $v\in S$ and every $i\in\{1,\ldots,p\}$, we introduce an integer-valued variable $x_{vi}$ meaning that exactly $x_{vi}$ neighbors of $v$ in $L_i\setminus Y_i$ receive the message from $v$. We have the following straightforward constraints:
\begin{equation}\label{eq:one}
x_{vi}\geq 0\text{ for all }v\in S\text{ and }i\in\{1,\ldots,p\}
 \end{equation}
 and 
 \begin{equation}\label{eq:two}
x_{vi}= 0\text{ for all }v\in S\text{ and }i\in\{1,\ldots,p\}\text{ s.t. }v\text{ is not adjacent to the vertices of }L_i.
 \end{equation}
 To encode that all vertices receive the message, we impose the following constraint:
 \begin{equation}\label{eq:three}
 \sum_{v\in S}x_{vi}=|L_i|-\ell_i\text{ for every }i\in\{1,\ldots,p\}.
 \end{equation}
 Finally, we encode the property that the number of rounds should be bounded by $t-2k$:
  \begin{equation}\label{eq:four}
 \sum_{i=1}^px_{vi}\leq t-2k\text{ for every }v\in S.
 \end{equation} 
 
 Observe that the number of variables in the system (\ref{eq:one})--(\ref{eq:four}) is at most $k+p\leq k+2^k$. Therefore, we can solve it in $2^{\Oh(k2^k)}\cdot n^{\Oh(1)}$ time by Lemma~\ref{lem:ILP}. If the algorithm from Lemma~\ref{lem:ILP} reports that the system   (\ref{eq:one})--(\ref{eq:four}) has a solution, we conclude that $(G,s,t)$ is a yes-instance and stop. If the system has no solution for all choices of $\ell_1,\ldots,\ell$ and corresponding feasible sets $X$, we report that $(G,s,t)$ is a no-instance.

To show correctness, we follow the description of the algorithm. We can assume that $(G,s,t)$ is an instance of \probTB with $k\geq 1$, $t>2k^2$, $|V(G)|\geq 2$, and 
$|I|>2k^2$ for $I=V(G)\setminus S$, because otherwise correctness follows directly from the description. 

Suppose that $(G,s,t)$ is a yes-instance. By Lemmas~\ref{lem:vc} and \ref{lem:is}, $G$ and $s$ have an optimal broadcasting protocol   such that the set of vertices $X$ that receive the message in the first $2k$ rounds satisfies the properties (i) $S\subseteq X$ and (ii) for $Y=S\cap I$, $|Y|\leq 2k^2$. For each $i\in\{1,\ldots,p\}$, let $\ell_i=L_i\cap Y$. We have that 
$\ell_1,\ldots,\ell_p$ are nonnegative integers such that
(a) $\ell_i\leq |L_i|$ for each $i\in\{1,\ldots,p\}$, (b) if $s\in L_i$ for some $i\in\{1,\ldots,p\}$, then $\ell_i\geq 1$,  and (c) $\ell_1+\dots+\ell_p\leq 2k^2$. Therefore, we consider the choice of $\ell_1,\ldots,\ell_p$ in one of the branches of our algorithm. Because the vertices of each $L_i$ are indistinguishable, we obtain that for our choice of $Y_1,\ldots,Y_p$, there is an  
 optimal broadcasting protocol $P$ such that the set of vertices $X$ that receive the message in the first $2k$ rounds is exactly $X'=S\cup(Y_1\cup\dots\cup Y_p)$. In particular, this means that $X'$ is feasible. Hence, the choice of $\ell_1,\ldots,\ell_p$ and $X'$ is not discarded by the algorithm. For each $v\in S$ and $i\in\{1,\ldots,p\}$, let $x_{vi}^*$ be the number of vertices of $L_i$ that receive the message from $v$ in the rounds from $2k+1$ to $t$ with respect to the protocol $P$. We obtain that these values of $x_{vi}^*$ satisfy (\ref{eq:one})--(\ref{eq:four}). Therefore, the system (\ref{eq:one})--(\ref{eq:four}) is feasible and our algorithm returns a correct yes-answer. 
 
 For the opposite direction, assume that there is a choice of  nonnegative integers $\ell_1,\ldots,\ell_p$ satisfying (a)--(c) such that for the corresponding choice of $Y_1,\ldots,Y_p$, 
 $X=S\cup(Y_1,\ldots,Y_p)$ is feasible, such that  the system (\ref{eq:one})--(\ref{eq:four})  has a solution. Suppose that  the values  $x_{vi}^*$ for $v\in S$ and $i\in\{1,\ldots,p\}$ give a solution of the system. We have that there is a broadcasting protocol $P$ that ensures that the vertices of $X$ receive the message in the first $2k$ rounds. We extend the protocol by defining that exactly $x_{vi}^*$ vertices of $L_i$ receive the message from each $v\in S$. Because of  (\ref{eq:one})--(\ref{eq:four}), we obtain that all the vertices get the message in $t$ rounds. Thus, $(G,s,t)$ is a yes-instance. This completes the proof of correctness.
 
To evaluate the running time, notice that we consider $2^{\Oh(k^3)}$ choices of   $\ell_1,\ldots,\ell_p$. For each choice, we either discard it in $2^{\Oh(k^2)}$ time or solve the system  
(\ref{eq:one})--(\ref{eq:four}) in $2^{\Oh(k2^k)}\cdot n^{\Oh(1)}$. This implies that the total running time is  $2^{\Oh(k2^k)}\cdot n^{\Oh(1)}$. This completes the proof of Theorem~\ref{thm:vc}.

\section{Kernelization for the parameterization by $k=n-t$}\label{sec:kernel}
In this section, we prove Theorem~\ref{thm:kernel}. Recall that we parameterize  \probTB by $k=n-t$. Hence, it is convenient for us to denote the considered instances as triples $(G,s,k)$ throughout the section instead of $(G,s,n-k)$. 

Let $(G,s,k)$ be an instance of \probTB. 
We exhaustively apply the following reduction rules in the order in which they are stated. The first rule is straightforward because $b(G,s)\leq n-1$ and $(G,s,k)$ is a yes-instance if $k\leq 1$.  
Also if $k>n$, then $(G,s,k)$ is a no-instance.

\begin{reduction}\label{rule:trivial}
If $k\leq 1$, then return a trivial yes-instance, e.g., the instance with $G=(\{s\},\emptyset)$ and $k=0$, and stop. If $k>n$, then return a trivial no-instance, e.g., the instance with  $G=(\{s,v\},\{sv\})$ and $k=1$, and stop.
\end{reduction}

Observe that after applying Reduction Rule~\ref{rule:trivial}, $n\geq 2$, because if $n=1$, then either $k\leq 1$ of $k>n$ and we would stop.
Notice that if $d_G(s)=1$, then the source $s$ sends the message to its unique neighbor in the first round. This allows us to delete $s$ and define a new source using the following rule whose safeness is straightforward. 

\begin{reduction}\label{rule:reroot}
If $d_G(s)=1$, then let $v$ be the neighbor of $s$, set $G:=G-s$ and define $s:=v$.
\end{reduction}

Now we can assume that $d_G(s)\geq 2$. By the next rule, we delete certain pendent vertices. 

\begin{reduction}\label{rule:pend}
If there is a vertex $v\in V(G)$ such that for the set of vertices of degree one $W\subseteq N_G(v)$, it holds that $|W|\geq |V(G)\setminus W|$, then select an arbitrary $w\in W$ and 
set $G:=G-w$. 
\end{reduction}

\begin{claim}\label{cl:pend}
Reduction Rule~\ref{rule:pend} is safe.
\end{claim}

\begin{proof}[Proof of Claim~\ref{cl:pend}]
Let $v\in V(G)$ and assume the set of vertices $W$ of degree one in the neighborhood of $v$ satisfies the condition $|W|\geq|V(G)\setminus W|$. Let also $w\in W$ and  $G'=G-w$. 
Notice that because $d_G(s)\geq 2$, $w\neq s$. Since $d_G(w)=1$, $w$ receives the message from $v$ and  $w$ does not send the message anywhere.  This implies that 
$b(G,s)-1 \leq b(G',s)\leq b(G,s)$. To show the claim, it is sufficient to prove that $b(G',s)=b(G,s)-1$. 

Consider an optimal broadcasting protocol $(T,\{C(x)\mid x\in V(T)\})$ for $G$. Because the vertices of $W$ are adjacent to $v$ and have degree one, the vertices of  $W$ are children of $v$ in $T$.  
Since each vertex $x\in W$ has no children in $T$, we can assume without loss of generality that the vertices of $W$ are the last vertices in $C(v)$ and, furthermore, $w$ is the last vertex in this ordered set. Then the vertices of $V(G)\setminus W$ get the message in at most $|V(G)|-|W|-1$ rounds. Because $|V(G)\setminus W|\leq |W|$, we obtain that $w$ gets the message in the last round. Moreover,  $w$ is a unique vertex that gets the message in the last round. This implies that $(T',\{C'(x)\mid x\in V(T')\})$, where $T'=T-w$, $C'(x)=C(x)$ for $x\in V(T')\setminus \{v\}$ and $C'(v)=C(v)\setminus\{w\}$, is a broadcasting protocol for $G'$ that ensures that every vertex gets the message in $b(G,s)-1$ rounds. Thus $b(G',s)=b(G,s)-1$. 
\end{proof}

To state the following rule, we introduce an auxiliary notation. For a vertex $v$ of a graph $H$, we define $\rho_H(v)=\max\{\dist_H(v,u)\mid u\in V(H)\}$. 

\begin{reduction}\label{rule:bridge}
If $G$ has a bridge $e=uv$ such that $G-e$ has two connected components $G_1$ and $G_2$, where $s,u\in V(G_1)$, $v\in V(G_2)$, $d_G(u)=2$, and $|V(G_1)|< \dist_{G_1}(s,u)+\rho_{G_2}(v)$, then set $G:=G/e$. 
\end{reduction}

\begin{claim}\label{cl:bridge}
Reduction Rule~\ref{rule:bridge} is safe.
\end{claim}

\begin{proof}[Proof of Claim~\ref{cl:bridge}]
Let $e=uv$ be a bridge of $G$ such that $G-e$ has two connected components  $G_1$ and $G_2$, where $s,u\in V(G_1)$, $v\in V(G_2)$, and $|V(G_1)|< \dist_{G_1}(s,u)+\rho_{G_2}(v)$. Let also $G'=G/e$. 
Because to reach the vertices of $G_2$, $u$ should send the message to $v$, we have that $b(G,s)-1 \leq b(G',s)\leq b(G,s)$. To prove the claim, we show that $b(G',s)=b(G,s)-1$. 

Let $(T,\{C(x)\mid x\in V(T)\})$ be an optimum broadcasting protocol for $G$. Because $e$ is a bridge, $v$ is a child of $u$ in $T$. Also for every $x\in V(G_2)$, $C(x)$ contains only vertices of $G_2$.
Because $u$ is a unique vertex in $V(G_1)$ which sends the message outside $V(G_1)$ and does it only once, we have that every vertex of $V(G_1)$ receives the message in at most $V(G_1)$ rounds. On the other side, we can observe that
$V(G_2)$ has a vertex that  gets the message only in at least $\dist_{G_1}(s,u)+\rho_{G_2}(v)$ rounds. Therefore, the vertices that receive the message in the last round are in $G_2$.
We define the protocol $(T',\{C'(x)\mid x\in V(T')\})$ for $G'$ as follows. We set $T'=T/e$, and for every $x\in V(T)\setminus \{u,v\}$, we define $C'(x)=C(x)$. Let $w$ be the vertex obtained from $u$ and $v$ by the contraction of $e$. To construct $C'(w)$, note that  $C(u)=\{v\}$, because $d_G(u)=2$. We define $C'(w)=C(v)$. 
Observe that by the protocol for $G'$, each vertex  $x\in V(G_2)$ gets the message in one round earlier than in the protocol for $G$. We conclude that each vertex of $G$ receives the message in at most $b(G,s)-1$ rounds. This means that $b(G',s)=b(G,s)-1$. 
\end{proof}

From now, we can assume that Reduction Rules~\ref{rule:trivial}--\ref{rule:bridge} are not applicable. We run the standard breadth-first search (BFS) algorithm on $G$ from $s$ (see, e.g.,~\cite{CormenLRS09} for the description). The algorithm produces a spanning tree $B$ of $G$ of shortest paths and  the partition of $V(G)$ into \emph{BFS-levels} $L_0,\ldots,L_r$, where $L_i$ is the set of vertices at distance $i$ from $s$ for every $i\in\{1,\ldots,r\}$.

We apply the following rule whose safeness immediately follows from Observation~\ref{obs:span}  and Lemma~\ref{lem:trees}.

\begin{reduction}\label{rule:bfs}
Compute $b(B,s)$ and if $b(B,s)\leq n-k$, then return a trivial yes-instance and stop.
\end{reduction}

Then we apply the final rule.

\begin{reduction}\label{rule:final}
If there is $v\in L_i$ for some $i\in \{0,\ldots,r-1\}$ such that for $X=N_G(v)\cap L_{i+1}$ and for the $(s,v)$-path $P$ in $B$, it holds that 
(i) $|X|\geq 2k+1$ and  
(ii) the total number of vertices in nontrivial, i.e., having at least two vertices, connected components of $G-V(P)$ containing vertices of $X$ is at least $4k-2$,
then return a trivial yes-instance and stop.
\end{reduction}

\begin{claim}\label{cl:final}
Reduction Rule~\ref{rule:final} is safe.
\end{claim}

\begin{proof}[Proof of Claim~\ref{cl:final}]
Let  $v\in L_i$ for some $i\in \{0,\ldots,r-1\}$ and let $P$ be the $(s,v)$-path in $B$. 
Denote by $G_1,\ldots,G_\ell$ the nontrivial connected components of $G-V(P)$ such that each of them contains at least one vertex of $X=N_G(v)\cap L_{i+1}$. 
We assume that $|V(G_1)|\geq \dots\geq |V(G_\ell)|\geq 2$.
Suppose that $|V(G_1)|+\cdots+|V(G_\ell)|\geq 4k-2$. To show that the rule is safe, we have to prove that $b(G,s)\leq n-k$. 
By Observation~\ref{obs:sub}, it suffices to show that there is a tree subgraph $T$ of $G$ containing $s$ with $b(T,s)\leq |V(T)|-k$.
We consider two cases.  

Assume first that $V(G_1)\geq k+1$. Let $u\in V(G_1)\cap X$. We find a tree subgraph $F$ in $G_1$ with exactly $k+1$ vertices such that $u\in V(F)$. Then to construct the tree $T$, we take $P$ and $F$ and make $v$ adjacent to $u$.  Because $|X|\geq 2k+1$, there are $k$ distinct vertices $u_1,\ldots,u_k\in X\setminus V(F)$. We include them in $T$ by making them adjacent to $v$. We define the broadcasting protocol  $(T,\{C(x)\mid x\in V(T)\})$ as follows. For each $x\in V(P)\setminus \{v\}$, $x$ has a unique child composing $C(x)$. We define 
$C(v)=(v,u_1,\ldots,u_k)$ and for each $x\in V(F)$, $C(x)$ is an arbitrary ordering of the children of $x$ in $F$. Because $V(F)=k+1$, we have that 
$b(T,s)\leq |V(P)|+|V(F)|-1\leq |V(T)|-k$. 

Suppose now that $k\geq |V(G_1)|\geq\dots\geq |V(G_\ell)|$. 
Let $T_i$ be a spanning tree of $G_i$ and let $u_i\in V(T_i)$ for each $i\in\{1,\ldots,\ell\}$. We construct $T$ from $P$ and $T_1,\ldots,T_\ell$ by making $v$ adjacent to $u_1,\ldots,u_\ell$.  
We define the broadcasting protocol  $(T,\{C(x)\mid x\in V(T)\})$ as follows. For each $x\in V(P)\setminus \{v\}$, $x$ has a unique child composing $C(x)$. We define 
$C(v)=(u_1,\ldots,u_\ell)$ and for each $i\in\{1,\ldots,\ell\}$ and $x\in V(T_i)$, $C(x)$ is an arbitrary ordering of the children of $x$ in $T_i$. 
Observe that
\begin{equation*}
b(T,s)\leq |V(P)|-1+\max\{|V(T_i)|+i-1\mid i\in\{1,\ldots,\ell\}\}. 
\end{equation*}
Since $|V(T_i)|\leq k$ for $i\in\{1,\ldots,k\}$, $b(T,s)\leq |V(P)|+k+\ell-2$. 
If $\ell\geq 2k$, then $|V(G_1)|+\cdots+|V(G_\ell)|\geq 2\ell$ and 
$|V(T)|=|V(P)|+|V(G_1)|+\cdots+|V(G_\ell)|\geq |V(P)|+\ell+2k$. Hence,
$|V(T)|-b(T,s)\geq k$. If $\ell\leq 2k$, then $b(T,s)\leq |V(P)|+3k-2$
and $|V(T)|-b(T,s)\geq (4k-2)-(3k-2)\geq k$.
This completes the proof.
\end{proof} 

The crucial property of the instance obtained by applying Reduction Rules~\ref{rule:trivial}--\ref{rule:final} is given in the following lemma.

\begin{lemma}\label{lem:bound}
Suppose that Reduction Rules~\ref{rule:trivial}--\ref{rule:final} are not applicable to $(G,s,k)$. Then $|V(G)|\leq 18k-12$.
\end{lemma}

\begin{proof}
Recall that $B$ is a BFS-tree rooted in $s$ and $b(B,s)> n-k$ because of Reduction Rule~\ref{rule:bfs}. Recall also that $L_0,\ldots,L_r$ are the BFS-levels.
 For a vertex $v\in V(B)$, we use $B_v$ to denote the subtree of $B$ induced by the descendants of $v$ in $B$ including $v$ itself. If $|V(G)|\leq k+1$, then the claim holds. Assume that 
 $|V(G)|\geq k+1$.
 Because $|V(B)|=|V(G)|\geq k+1$, there is $v\in V(B)$ such that $|V(B_v)|\geq k+1$ but $|V(B_u)|\leq k$ for every child $u$ of $v$ in $B$. 
Let $P$ be the $(s,v)$-path in $B$. We define  $Y=V(G)\setminus (V(P)\cup V(B_v))$.  Clearly, $|V(G)|=|V(B)|=|Y|+|V(P)|+|V(B_v)|-1$. To prove the lemma, we show upper bounds for $|Y|$, $|V(P)|$, and $|V(B_v)|$. First, we show a bound for the size of $Y$.

\begin{claim}\label{cl:Y}
$|Y|\leq k-1$.
\end{claim}

\begin{proof}[Proof of Claim~\ref{cl:Y}]
For the sake of contradiction, assume that $|Y|\geq k$. 
Let $P=v_1\cdots v_p$, where $s=v_1$ and $v=v_p$. 
We consider the following broadcasting protocol $(B,\{C(x)\mid x\in V(B)\})$ for $B$. For every $i\in \{1,\ldots,p-1\}$, we construct $C(v_i)$ by making $v_{i+1}$ the first element of the ordered set and then append the other children of $v_i$ in $B$ in arbitrary order. For $x\in V(B)\setminus\{v_1,\ldots,v_{p-1}\}$, $C(x)$ is an arbitrary ordering of the children of $x$.  Notice that the vertices of $P$ are getting the message in the first $|V(P)|-1$ rounds. Then because $|Y|\geq k$ and each vertex $y\in Y$ is reachable in $B-V(B_v)$ from $V(P)\setminus \{v\}$ by a $(v_i,y)$-path for some $i\in\{1,\ldots,p\}$, at least $k$ vertices of $Y$ should receive the message in the first $|V(P)|+k-1$ rounds.  Similarly, because $|V(B_v)|\geq k+1$, at least $k$ vertices of $B_v-v$ should receive the message in the first $|V(P)|+k-1$ rounds. Hence, at least $|V(P)|+2k$ vertices of $B$ get the message in the first  $|V(P)|+k-1$ rounds. This implies that the total number of rounds for the protocol $(B,\{C(x)\mid x\in V(B)\})$ is at most $|V(B)|-k$. This contradicts the assumption that Reduction Rule~\ref{rule:bfs} is not applicable and proves the claim. 
\end{proof}

Next, we use Claim~\ref{cl:Y} to upper bound the number of vertices of $P$.

\begin{claim}\label{cl:P}
$|V(P)|\leq 4k-2$.
\end{claim}

\begin{proof}[Proof of Claim~\ref{cl:P}]
The proof is by contradiction. Assume that $|V(P)|\geq 4k-1$. 
 Note that $v\in L_p$ for $p\geq 4k-2$ and $r\geq p+1\geq 4k-1$. 
 By Claim~\ref{cl:Y}, $|Y|\leq k-1$. Notice that $d_G(s)\geq 2$, because Reduction Rule~\ref{rule:reroot} is not applicable.
 By the pigeonhole principle, there is $q\leq 3k-3$ such that $|L_{q-1}|=L_{q}|=|L_{q+1}|=1$, that is, $L_{q-1}=\{v_{q-1}\}$, $L_{q}=\{v_q\}$, and $L_{q+1}=\{v_{q+1}\}$. This implies that  $d_G(v_q)=2$ and $v_{q}v_{q+1}$ is a bridge of $G$.  Consider the connected components $G_1$ and $G_2$ of $G-v_{q}v_{q+1}$ and assume that $s,v_{q}\in V(G_1)$ and $v_{q+1}\in V(G_2)$.\ We have that $|V(G_1)|\leq q+1+|Y|\leq 4k-3$. It also holds that $\dist_{G_1}(s,v_{q})=q$ and $\rho_{G_2}(v_{q+1})=r-q-1$. 
 Hence, 
 $\dist_{G_1}(s,v_q)+\rho_{G_2}(v_{q+1})=r-1\geq 4k-2$. Thus, $|V(G_1)|\leq 4k-3< \dist_{G_1}(s,v_q)+\rho_{G_2}(v_{q+1})$ but this implies that Reduction Rule~\ref{rule:bridge} would be applicable for the bridge $v_qv_{q-1}$. This gives a contradiction that proves the claim.   
\end{proof}

Finally, we upper bound the size of $B_v$.

\begin{claim}\label{cl:Bv}
$|V(B_v)|\leq 13k-9$.
\end{claim}

\begin{proof}
Assume that $v\in L_p$ for some $p\in\{1,\ldots,r\}$ and denote by $Z\subseteq X=N_G(v)\cap L_{p+1}$ the set of vertices having of degree one. We consider two cases depending on the size of $X$. 

Suppose that $|X|\geq 2k+1$. Because of Reduction Rule~\ref{rule:final},  the total number of vertices in nontrivial, i.e., having at least two vertices, connected components of $G-V(P)$ containing vertices of $X$ is at most $4k-3$. Hence, $B_v-Z$ has at most $4k-2$ vertices. Because $|Y|\leq k-1$ and $|V(P)|\leq 4k-2$ by Claims~\ref{cl:Y} and \ref{cl:P}, respectively, 
$B-Z$ has at most $(k-1)+(4k-2)+(4k-2)-1=9k-6$ vertices, that is, $|V(G)\setminus Z|\leq 9k-6$. Because Reduction Rule~\ref{rule:pend} is not applicable, we have that 
$|Z|\leq |V(G)\setminus Z|-1\leq 9k-7$. We obtain that 
$|V(B_v)|=|V(B_v)\setminus Z|+|Z|\leq 13k-9$. 

Assume now that $|X|\leq 2k$. Suppose that $T_1,\ldots,T_\ell$ are the subtrees of $B_v$ rooted in the children $u_1,\ldots,u_\ell$ of $v$, respectively, such that $|V(T_i)|\geq 2$ for each $i\in\{1,\ldots,\ell\}$.  Notice that $\ell\leq |X|\leq 2k$ and $|V(T_i)|\leq k$ because of the choice of $v$. We claim that $|V(T_1)|+\dots+|V(T_\ell)|\leq 4k-2$.

The proof is by contradiction. Let $|V(T_1)|+\dots+|V(T_q)|\geq 4k-1$. Consider the tree $T$ constructed from $P$ and $T_1,\ldots,T_\ell$ by making $v$ adjacent to $u_1,\ldots,u_\ell$.  
We show that $b(T,s)\leq |V(T)|-k$. For this,  we define the broadcasting protocol  $(T,\{C(x)\mid x\in V(T)\})$ as follows. For each $x\in V(P)\setminus \{v\}$, $x$ has a unique child composing $C(x)$. We define 
$C(v)=(u_1,\ldots,u_\ell)$ and for each $i\in\{1,\ldots,\ell\}$ and $x\in V(T_i)$, $C(x)$ is an arbitrary ordering of the children of $x$ in $T_i$. 
We have that
\begin{equation*}
b(T,s)\leq |V(P)|-1+\max\{|V(T_i)|+i-1\mid i\in\{1,\ldots,\ell\}\}. 
\end{equation*}
Since $|V(T_i)|\leq k$ for $i\in\{1,\ldots,k\}$ and $\ell\leq 2k$, $b(T,s)\leq |V(P)|+k+\ell-2\leq |V(P)|+3k-1$.
Because $|V(T)|=|V(P)|+|V(T_1)|+\dots+|V(T_q)|\geq |V(P)|+4k-1$, we obtain that 
$b(T,s)\geq |V(T)|-k$. However, $T$ is a subgraph of the tree $B$ and $s\in V(T)$. Then by Observation~\ref{obs:sub},
 $b(B,s)\leq b(T,s)+|V(B)\setminus V(T)|\leq |V(B)|-k=n-k$ contradicting that Reduction Rule~\ref{rule:bfs} is not applicable. This proves that $|V(T_1)|+\dots+|V(T_\ell)|\leq 4k-2$.
 Because $|V(B_v)|= |X|-\ell+|V(T_1)|+\dots+|V(T_\ell)|+1$, we obtain that  $|V(B_v)|\leq |X|+|V(T_1)|+\dots+|V(T_\ell)|\leq 2k+4k-2\leq 6k-2$. 
This completes the proof of the claim.
\end{proof}

Summarizing the upper bounds from Claims~\ref{cl:Y}--\ref{cl:Bv}, we derive that 
\[|V(G)|=|Y|+|V(P)|+|V(B_v)|-1\leq (k-1)+(4k-2)+(13k-9)=18k-12.\] This concludes the proof.
\end{proof}

By Lemma~\ref{lem:bound}, if we do not stop during the exhaustive applications of Reduction Rules~\ref{rule:trivial}--\ref{rule:final}, then for the obtained instance $(G,s,k)$, $|V(G)|\leq 18k-12$. Hence, to complete the kernelization algorithm, we return $(G,s,k)$.

It is straightforward to see that Reduction Rules~\ref{rule:trivial}--\ref{rule:final} can be applied in polynomial time. In particular, BFS and finding bridges can be done in linear time by classical graph algorithms (see, e.g., the textbook~\cite{CormenLRS09}). Thus, the total running time of the kernelization algorithm is polynomial. This completes the proof of Theorem~\ref{thm:kernel}.

\section{Conclusion}\label{sec:concl}
In our paper, we initiated the study of \probTB from the parameterized complexity viewpoint. In this section, we discuss further directions of research. 

We observed that \probTB is trivially \classFPT when parameterized by $t$ and Theorem~\ref{thm:exact} implies that the problem can be solved  in $3^{2^t}\cdot n^{\Oh(1)}$ time. Is it possible to get a better running time for the parameterization by $t$?

In Theorem~\ref{thm:kernel}, we obtained a polynomial kernel for the parameterization by $k=n-t$, that is, for the parameterization below the trivial upper bound for $b(G,s)$. This naturally leads to the question about parameterization below some other bounds for this parameter. We note that the parameterization of \probTB above the natural lower bound $b(G,s)\geq \log n$ leads to a para-NP-complete problem.  
To see this, observe that for graphs with $n=2^t$ vertices, $b(G,s)\leq t$ if and only if $G$ has a binomial spanning tree rooted in $s$, and it is NP-complete to decide whether $G$ contains such a spanning tree~\cite{PapadimitriouY82}.

In Theorems~\ref{thm:cyc}  and \ref{thm:vc}, we considered structural parameterizations of \probTB
 by the cyclomatic and vertex cover numbers, respectively.  It is interesting to consider other structural parameterizations. In particular, is \probTB \classFPT when parameterized by the \emph{feedback vertex number} and \emph{treewidth} (we refer to~\cite{CyganFKLMPPS15} for the definitions)? For the parameterization by treewidth, the complexity status of \probTB  is open even for the case when the treewidth of the input graphs is at most two, that is, for series-parallel graphs.

\end{document}